\documentclass[conference]{IEEEtran}
\usepackage{setspace}
\usepackage{etex}
\usepackage{booktabs}
\usepackage{graphicx}
\usepackage[cmex10]{amsmath}
\usepackage{amsfonts}
\usepackage{amssymb}
\usepackage{mathrsfs}
\usepackage{bm}
\usepackage{algorithmic}
\usepackage{algorithm}
\usepackage{array}
\usepackage{mdwmath}
\usepackage{mdwtab}
\usepackage{stfloats}
\usepackage{footnote}
\usepackage{tabularx}
\usepackage{multirow}
\usepackage{color}
\usepackage{textcomp}
\usepackage{pspicture}
\usepackage{pst-all}
\usepackage{xcolor}
\usepackage{pdftricks}
\usepackage{pstricks-add}
\usepackage{tikz}
\usepackage{amsthm}
\usetikzlibrary{arrows}
\usepackage{changes}
\usepackage[caption=true,font=footnotesize,labelfont=sf,textfont=sf]{subfig}

\ifCLASSINFOpdf
\else
\fi

\newtheorem{theorem}{Theorem}

\newtheorem{remark}{Remark}

\newtheorem*{theorem*}{Theorem}

\IEEEoverridecommandlockouts

\begin{document}

\title{Dynamic Edge Caching with Popularity Drifting}

\author{\IEEEauthorblockN{Linqi Song}
\IEEEauthorblockA{Department of Computer Science\\
City University of Hong Kong, Hong Kong SAR}
\and
\IEEEauthorblockN{Jie Xu\\}
\IEEEauthorblockA{Department of Electrical and Computer Engineering\\
University of Miami, USA}}

\maketitle

\begin{abstract}
Caching at the network edge devices such as wireless caching stations (WCS) is a key technology in the 5G network. The spatial-temporal diversity of content popularity requires different content to be cached in different WCSs and periodically updated to adapt to temporal changes. In this paper, we study how the popularity drifting speed affects the number of required broadcast transmissions by the MBS and then design coded transmission schemes by leveraging the broadcast advantage under the index coding framework. The key idea is that files already cached in WCSs, which although may be currently unpopular, can serve as side information to facilitate coded broadcast transmission for cache updating. Our algorithm extends existing index coding-based schemes from a single-request scenario to a multiple-request scenario via a ``dynamic coloring'' approach. Simulation results indicate that a significant bandwidth saving can be achieved by adopting our scheme.
\end{abstract}

\section{Introduction}
Proactively caching popular bulky traffic (e.g. videos) in the network edge devices such as wireless caching stations (WCSs) or cache-enabled small cells is a promising approach to alleviate the backhaul bandwidth burden of the mobile network and reduce content access time \cite{shanmugam2013femtocaching}. Since content popularity among users is, to a certain extent, predictable, popular content can be pre-cached at the WCSs close to users before actual requests arrive. In a common scenario illustrated in Fig.~\ref{system}, WCSs are deployed in a ``drop-and-play'' manner without wired connections along roadside to enhance network capacity while conventional macro base stations (MBSs) provide ubiquitous coverage and control signalling \cite{zhang2017self}.

Content popularity varies both spatially and temporally. On the one hand, WCSs placed in different locations serve different users who may have different preferences over the content. Therefore what content to cache is likely to be different across WCSs. On the other hand, content popularity evolves over time as new content is being produced and hence, caches of the WCSs must be periodically refreshed to adapt to the temporal popularity changes. As more and more WCSs are being deployed at the network edge to provide ubiquitous and fast content access, the spatial-temporal diversity of content popularity begins to impose an increasingly heavy traffic burden on the wireless link between the MBS and the distributed WCSs, taking up precious wireless bandwidth of the network.

\begin{figure}
  \centering
  \includegraphics[width=0.45\textwidth]{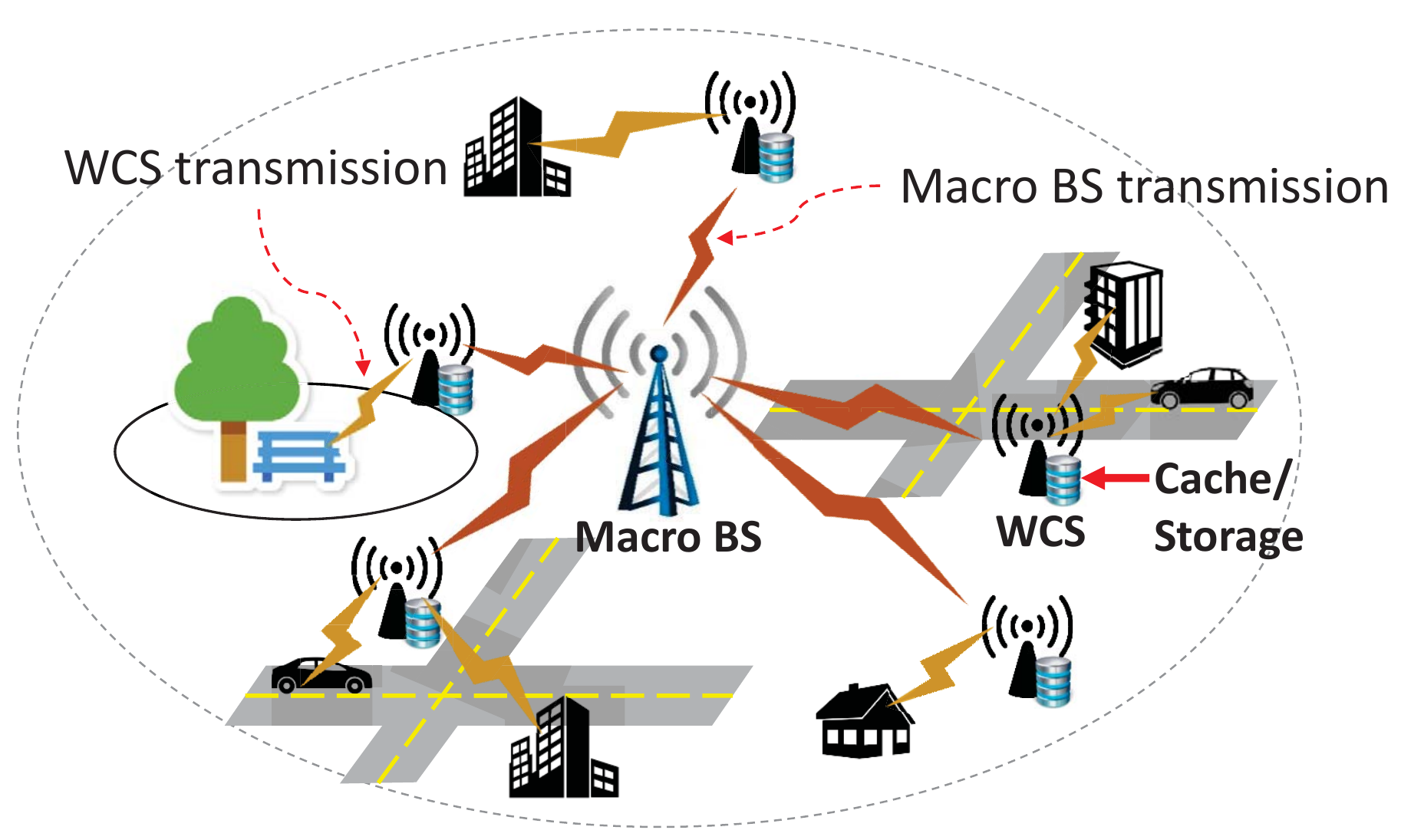}\\
  \caption{Network architecture with wireless caching stations}\label{system}
\end{figure}

In this paper, we study dynamic proactive caching among a network of distributed WCSs and design efficient transmission schemes to minimize the bandwidth usage for pushing popular content into WCSs from the MBS. Specifically, we design coded broadcast transmission schemes under the index coding framework. Our scheme is developed based on an observation: as content popularity evolves over time, the previously cached content, which although may not be popular at the current time, can serve as \textit{side information} to facilitate coded broadcast transmission among several WCSs, thereby saving wireless backhaul bandwidth. The main contributions are as follows: We model the popularity drifting of users as a dynamic process that is characterized by the distance measure between different rankings. The popularity drifting indicates that the preference rankings slightly differ between rounds. We formulate the transmission problem for dynamic proactive caching among a network of WCSs, and show a consistent trend that the number of required transmissions increases with the content popularity drifting speed under both uncoded transmission schemes and coded transmission schemes. For uncoded transmissions, we show that in the worst case, the transmissions needed to refresh the cache is proportional to the drifting speed parameter. We design  MDS codes for proactive caching with side information and characterize its bandwidth savings. We further design the optimal transmission schemes under the framework of index coding.We leverage existing graph-coloring-based index coding schemes and propose improvements tailored to our problem, termed dynamic graph coloring. Furthermore, we show that our proposed coded transmission and caching schemes can save a fraction of $\frac{s}{n\sqrt{c}\log(s)}$ transmissions compared with the uncoded schemes, where $s$ is the cache size, $n$ is the number of WCSs, and $c$ is a drifting parameter.

\section{Related Work}
Caching at the network edge has recently attracted lots of attention. The concept of FemtoCaching was introduced in \cite{shanmugam2013femtocaching} which studies content placement at small cell BSs to minimize the content access delay. Geographical caching was investigated in \cite{blaszczyszyn2015optimal} to maximize the probability of serving a user. Coded caching is an emerging topic recently, especially in wireless broadcast channels \cite{maddah2014fundamental,pedarsani2016online} or Device-to-Device networks \cite{ji2016fundamental}. The main coding techniques used in coded caching stem from the index coding \cite{bar2011index}, but focus more on how to place the cached content. Index coding is shown to be NP-hard to approximate within a constant factor \cite{bar2011index,langberg2011hardness} and various herustics are proposed to realize the codes \cite{chaudhry2008efficient}. In \cite{maddah2014fundamental,pedarsani2016online,ji2016fundamental}, files are separable and the concern is how to place (fractions of) files in the distributed WCSs so that the number of coded transmissions by the MBS is minimized when the actual requests arrive. They seek to uncover the information-theoretic limits of caching. In our problem, we consider a more practical setting where files are not separable and content placement is governed by exogenous content popularity. Our focus is on how to design coded transmission schemes to minimize bandwidth usage given the content distribution pattern and finding the consistent trend of communication cost with respect to the popularity drifting.

\section{System Model}
Consider a wireless network with one macro base station (MBS) and $n$ wireless caching stations (WCS), denoted by the set $[n]\triangleq \{1,2,...,n\}$. The WCSs are distributed over the network and can receive data from the MBS via a wireless broadcast channel. For analytical simplicity, we assume that this broadcast channel is error-free in this paper. Each WCS $i \in [n]$ can proactively cache popular content from the remote server via the MBS, and deliver the content, when requested, to the end users in its wireless transmission range. By offloading the downlink traffic from MBS to the WCSs, which are in close proximity to the end users, proactive caching reduces transmission latency and relieves traffic burdens on the backhaul network. We consider a pool of $m$ files, denoted by $\mathcal{B} = \{b_1, b_2, ..., b_m\}$, at the remote server that can be cached in the WCSs. Without loss of generality, we assume that these files are of the same size. Files of different sizes can be divided into file chunks of equal size to satisfy this assumption. Each WCS has a cache of limited capacity that can store at most $s < m$ files. In some existing theoretical work \cite{maddah2014fundamental,pedarsani2016online,ji2016fundamental}, the WCS may only store parts of a file. However, for practical concerns, such as file management, we consider that the WCS can only cache an entire file. Because not all files can be cached in the WCS, which files to cache will be determined according to the file popularity among the users.

Time is divided into slots. At the beginning of each time slot $t$, each WCS $i$ estimates the popularity of each file among users in its coverage area, which may vary across different WCSs. Due to the limited cache capacity of a WCS, the $s$-most popular files have to be cached to maximize the caching performance. For the purpose of this paper, only the popularity \textit{ranking} over the $m$ files is relevant to our problem. Let $\pi_{t, i}: \mathcal{B} \to \{1,...,m\}$ be a ranking function with respect to WCS $i$ in time slot $t$, where $\pi_{t,i}(b)$ is the position or rank of file $b \in \mathcal{B}$. In addition, let $\mathcal{B}^s_{\pi}$ denote the set of top-$s$ files under a ranking $\pi$.

As file popularity, captured by the popularity ranking $\pi_{t,i}$, varies over time, cached files have to be refreshed at the beginning of every time slot $t$. However, because the already-cached files in the previous time slot $t-1$ may have overlap with the predicted top-$s$ popular files in the current time $t$, not all files need to be downloaded from the remote server via the MBS. These already-cached files in a specific time slot are termed \textit{side information} in the proactive caching problem. Moreover, since different WCSs may have different files cached in the previous time slot $t-1$, coding schemes can be designed to minimize the number of broadcast transmissions, thereby saving the backhaul bandwidth. We then ask how does the number of broadcast transmissions depend on the content popularity drifting speed over time.

\section{Uncoded Transmission for Proactive Caching}
First, we study proactive caching using a straightforward uncoded broadcast transmission. Let $S_{t,i}$ be the set of files cached in WCS $i$ in time slot $t$, which equals $\mathcal{B}^s_{\pi_{t,i}}$. For each WCS $i$, only the files that are in $S_{t,i}$ but not in $S_{t-1,i}$ need to be transmitted by the MBS to WCS $i$ to update its cached content. These files are denoted by $S_{t,i}\backslash S_{t-1,i}\triangleq R_{t,i}$. Since the MBS broadcasts files to all WCSs in the network, the files that need to be broadcasted is $\cup_{i\in [n]} S_{t,i}\backslash S_{t-1,i} \triangleq R_t$, and the total number of broadcast transmissions is $|\cup_{i\in [n]} S_{t,i}\backslash S_{t-1,i}|\triangleq T_{t, un}$. Clearly, the number of required transmissions depends on how fast the popularity ranking changes: if there is a dramatic change in the popularity ranking between consecutive time slots, then it is likely that more transmissions are needed. 

We first introduce some concepts regarding popularity ranking. The dissimilarity between two popularity rankings $\pi_1$ and $\pi_2$ is characterized by their distance, under metrics such as the Spearman footrule distance and the Kendall tau distance \cite{dwork2001rank} among others. In this paper, we adopt the Kendall tau distance metric, which is defined as the number of pair-wise differences between two rankings. This can be seen as a ``bubble sort'' distance, which is the number of pair-wise adjacent transpositions needed to sort one ranking to another. Let $K(\pi_1, \pi_2)$ denote the Kendall tau distance between two popularity rankings, which is formally defined as follows
\begin{align}
&K(\pi_1, \pi_2) = \\
&\left|\{(j_1, j_2) : j_1 \neq j_2, \pi_1(b_{j_1}) < \pi_1(b_{j_2}), \pi_2(b_{j_1}) > \pi_2(b_{j_2})\}\right| \nonumber
\end{align}

\textbf{Example}: \textit{Consider 4 files $\{b_1,b_2,b_3,b_4\}$. Assume that the first popularity ranking is $\pi_1(b_1) = 1$, $\pi_1(b_2) = 2$, $\pi_1(b_3) = 3$, $\pi_1(b_4) = 4$. Hence, file $b_1$ is the most popular. Assume that the second popularity ranking is $\pi_1(b_1) = 3$, $\pi_1(b_2) = 4$, $\pi_1(b_3) = 1$, $\pi_1(b_4) = 2$. Hence, file $b_3$ is the most popular. In order to calculate the Kendall tau distance, pair each file with every other file and count the number of times the values in ranking $\pi_1$ are in the opposite order of the values in ranking $\pi_2$. For instance, for the pair $(b_1, b_2)$, $\pi_1$ and $\pi_2$ are consistent because $\pi_1(b_1) < \pi_1(b_2)$ and $\pi_2(b_1) < \pi_2(b_2)$. However, for the pair $(b_1, b_3)$, the two rankings are inconsistent because $\pi_1(b_1) < \pi_1(b_3)$ whereas $\pi_2(b_1) > \pi_2(b_3)$. Among all six possible pairs, pairs $(b_1, b_3)$, $(b_1, b_4)$, $(b_2, b_3)$, $(b_2, b_4)$ make the two rankings inconsistent. Therefore, the Kendall tau distance between these two rankings is $K(\pi_1, \pi_2) = 4$.}

To characterize the popularity drift over time, we assume that, for two consecutive time slots $t - 1$ and $t$, the popularity ranking differ at most $c$, i.e. $K(\pi_{t-1,i}, \pi_{t, i}) \leq c$ for all time slot $t$ and WCS $i$. Therefore, the constant $c$ sets an upper bound on the speed of popularity drifting. The following theorem characterizes the relationship between the number of required transmissions $T_{un}$ and the popularity drifting speed $c$.

\begin{theorem}
With uncoded transmission, at the beginning of time slot $t$, the system needs
\begin{enumerate}
  \item at least one transmission, if for some WCS $i$, $K(\pi_{t-1,i}, \pi_{t,i}) > \frac{s(s-1)(m-s)(m-s-1)}{4}$.
  \item at most $\min\{n\sqrt{c}, m\}$ transmissions, if for every WCS $i$, $K(\pi_{t-1,i}, \pi_{t,i}) \leq c$.
\end{enumerate}
\end{theorem}

\begin{proof}
We first prove the first part of the theorem. Since using uncoded transmission scheme, for some node $i \in [n]$, if $S_{t,i}\backslash S_{t-1,i} \not= \emptyset$, then the system will need at least one transmission. Observe that if the two rankings $\pi_{t,i}$ and $\pi_{t-1,i}$ have the same set of top $s$ ranked files, i.e., $S_{t,i} = S_{-1t,i}$, then their distance can be at most ${s \choose 2}{m-s \choose 2}$. Therefore, a sufficient condition for $S_{t,i}\backslash S_{t-1,i} \not= \emptyset$ is that the distance between the two rankings $\pi_{t,i}$ and $\pi_{t-1,i}$ exceeds $\frac{s(s-1)(m-s)(m-s-1)}{4}$.

Next, we prove the second half of the theorem. It suffices to show that $|S_{t,i}\backslash S_{t-1,i}| \le \sqrt{c}$. Let us denote by $\Delta_1$ the set $S_{t,i}\backslash S_{t-1,i}$ and by $\Delta_2$ the set $S_{t-1,i}\backslash S_{t,i}$. Note that $|\Delta_1| = |\Delta_2|$, then it is not hard to see that the files indexed by $\Delta_1$ are ranked higher than the files indexed by $\Delta_2$ according to the ranking $\pi_{t,i}$, but the files indexed by $\Delta_1$ are ranked lower than the files indexed by $\Delta_2$ according to the ranking $\pi_{t-1,i}$. Therefore, the Kendall tau distance $K(\pi_{t,i},\pi_{t-1,i})$ is at least $|\Delta_1||\Delta_2| = |\Delta_1|^2 \le c$, indicating that $|\Delta_1| = |\Delta_2| \le \sqrt{c}$.
\end{proof}

From Theorem 1, we can see that \textit{in the worst case}, the number of broadcast transmissions needed may still be proportional to the number of WCSs due to the diversity in the files already cached in the WCSs. When there is a large number of WCSs, proactive caching consumes a significant amount of wireless backbone bandwidth. 

\section{Coded Transmissions for Proactive Caching}
We study the bandwidth-drifting relationship for coded broadcast transmission for proactive caching update, we first show this relationship for the Maximum Distance Separable (MDS) code and then for the index code.

\subsection{MDS Coding based Proactive Caching}
Using a $[\nu, \kappa]$ MDS code, we can encode $\kappa$ original files $b_1, ..., b_\kappa$ into $\nu$ encoded files $x_1, ..., x_\nu$, then we can decode the original $\kappa$ files by receiving \textit{any} $\kappa$ encoded files among the $\nu$ encoded ones.

Denote by $S^\dagger_{t-1,i}$ the set of cached files of WCS $i$ in time slot $t-1$ that are also in the request file set $R_t$, i.e., $S^\dagger_{t-1,i} = R _t\cap S_{t-1,i}$. The following theorem characterizes the number of broadcast transmissions needed to refresh the cache at the beginning of time slot $t$ using the MDS code.

\begin{theorem}
\label{thm:mds}
With MDS codes, the number of broadcast transmissions needed to refresh the cached content at the beginning of time round $t$ is at most $T_{t, un} - \min_{i \in [n]}\{|S^\dagger_{t-1,i}|\}$.
\end{theorem}

\begin{proof}
We will use a constructive proof method by designing the MDS coding scheme to refresh the cached content. We consider the following general encoded broadcast transmission scheme for time round $t$.
\begin{equation}
\left[
\begin{array}{cccc}
a_{11} & a_{12} & \ldots & a_{1m} \\
a_{21} & a_{22} & \ldots & a_{2m} \\
\vdots & \vdots & \ddots & \vdots \\
a_{T1} & a_{T2} & \ldots & a_{Tm}
\end{array}
\right]
\left[
\begin{array}{c}
b_{1} \\
b_{2} \\
\vdots \\
b_{m}
\end{array}
\right]
=
\left[
\begin{array}{c}
x_{1} \\
x_{2} \\
\vdots \\
x_{T}
\end{array}
\right],
\end{equation}
where $A=\{a_{\tau j}\}$ is the coding coefficient matrix; $T$ is the number of broadcast transmissions; and the $\tau$-th transmission is $x_\tau = a_{\tau 1}b_1 + a_{\tau 2} b_2 +\ldots + a_{\tau m} b_m$. This can also be written in the matrix form as $A\bm b =\bm x$, where $\bm b$ collects all the original files and $\bm x$ collects all the encoded transmissions. Obviously, we only need to transmit the files in $R_t$. Therefore, we can set the coding coefficients to $0$ corresponding to files $j \in [m] \backslash R_t$ without losing any transmission efficiency. This is equivalently to design a coding coefficient matrix $A^\dagger$ with only columns of $A$ corresponding to files in $R_t$. Thus we can write the encoding process as $A^\dagger \bm b^\dagger =\bm x$, where $\bm b^\dagger$ collects the original files indexed by $R_t$.

Now, we select the coding coefficient matrix $A^\dagger$ with $T = T_{un} - \min_{i \in [n]}\{|S^\dagger_{t-1,i}|\}$ such that any $T$ columns of all the $T_{un}$ columns are linearly independent. This can be obtained by the generator matrix of a $[T_{un},T]$ MDS code. After the coefficient matrix $A^{\dagger}$ is designed, it is commonly among the server and all caching stations. For WCS $i$, it can remove from the transmissions the part corresponds to files $S^\dagger_{t-1,i}$, i.e., for the $\tau$'s transmission, $x'_\tau \triangleq x_\tau - \sum_{j \in S^\dagger_{t-1,i}} a_{\tau j}b_j= \sum_{j\in R\backslash S^\dagger_{t-1,i}} a_{\tau j}b_j$. Therefore, the WCS $i$ knows the vector $\bm x'$ that collects all $x'_\tau$ and a matrix $A^{\dagger}_i \in \mathbb{F}^{T \times (T_{un}-|S^\dagger_{t-1,i}|)}_q$ that collects all columns corresponding to files in $R_t\backslash S^\dagger_{t-1,i}$; and then needs to solve the equation $A^{\dagger}_i \bm b^\dagger = \bm x'$ to get $\bm b^\dagger$. By our design of the transmission scheme, we have that any $T$ columns of the matrix $A^\dagger$ are linearly independent, and thus, having any $T_{un}-|S^\dagger_{t-1,i}| \le T_{un} - \min_{i \in [n]}\{|S^\dagger_{t-1,i}| =T$ columns linearly independent. Therefore, caching node $i$ can solve the equation $A^{\dagger}_i \bm b^\dagger = \bm x'$ (note that the variable is $\bm b^\dagger$ and the constant is $\bm x'$) to get a unique solution of $\bm b^\dagger$.
This is the case for all $i$ and then the cached content can be refreshed using at most $T_{un} - \min_{i \in [n]}\{|S^\dagger_{t-1,i}|\}$ number of broadcast transmissions.
\end{proof}

Theorem 2 shows that we can save at least a number of $\min_{i \in [n]}\{|S^\dagger_{t-1,i}|\}$ broadcast transmissions by using the MDS coding scheme compared to the uncoded transmission scheme. In particular, $\min_{i \in [n]}\{|S^\dagger_{t-1,i}|\}$ is bigger if the side information diversity is larger and hence, more savings can be achieved.

\subsection{Index Coding Based Proactive Caching}
We cast this problem as an index coding problem with side information where the side information is the already-cached files. In our problem, one feature is that each WCSs in each time slot may request {\emph {multiple}} files whereas in the conventional index coding problem, the schemes are designed often for single request. In this sense, we need to find algorithms that are efficient for multiple requests in order to achieve higher bandwidth efficiency.

The index coding problem has been shown to be NP-hard \cite{bar2011index}. The literature has shown that the index coding problem is hard to approximate within a constant ratio \cite{langberg2011hardness} and the existing algorithms are heuristics with either no theoretical bound of the approximation ratio or very loose upper bound \cite{chaudhry2008efficient,blasiak2010index}. In \cite{bar2011index}, the idea that the optimal linear index coding is upper bounded by the chromatic number of specifically defined ``conflict'' graph provides a good thread for designing index coding algorithms based on graph coloring. In this paper, we design our algorithms using graph theory based approach. 

Recall that we can reduce the multiple request case as multiple WCSs with single request who have the same side information \cite{bar2011index}. We explore the standard greedy coloring heuristics to find the chromatic number of the conflict graph for the obtained single-request index coding problem. The conflict graph $G = (V, E)$ \cite{bar2011index}, for this reduced single-request index coding problem is constructed as follows. Each vertex on this graph represents a virtual WCS, namely a WCS with one requested file. Therefore, there are totally $\sum_{i\in [n]} |R_{i,t}|$ vertices. Consider any two vertices $v_1$ and $v_2$, where $v_1 = (i_1, b_{j_1})$ represents WCS $i_1$ requesting file $b_{j_1}$ and $v_2 = (i_2, b_{j_2})$ represents WCS $i_2$ requesting file $b_{j_2}$. There is an edge between $v_1$ and $v_2$ if and only if $j_1\neq j_2$ and $j_1 \not\in R_{j_2, t} \lor j_2 \not\in R_{j_1, t}$.

To construct the broadcasting transmission scheme, we perform coloring on the conflict graph $G$. Each color will then correspond to a coded broadcast transmission. Indeed, it is not hard to see that if two vertices $v_1 = (i_1, b_{j_1})$ and $v_2 = (i_2, b_{j_2})$ have the same color, then either $j_1 = j_2$ or $j_1 \in R_{i_2,t} \land j_2 \in R_{i_1}$. Therefore, by transmitting either $b_{j_1} = b_{j_2}$ (for the case $j_1 = j_2$) or $b_{j_1} + b_{j_2}$ (for the case $j_1 \in R_{i_2,t} \land j_2 \in R_{i_1}$), WCS $i_1$ (or $i_2$) can decode $j_1$ (or $j_2$). We consider a greedy coloring method. Given an order of vertices $v_1, v_2, ..., v_{|V|}$ of a graph $G$, the greedy coloring operates across vertices: assign to vertex $v_1$ color 1; assign to vertex $v_2$ color 1 if vertex $v_2$ is not connected with vertex $v_2$ and color 2 otherwise; for the remaining vertex $v = v_3, ..., v_{|V|}$, assign the first available color. Let $order(v)$ be the ordered number of vertex $v$. In particular, there are two commonly used heuristic ordering methods.
\begin{itemize}
  \item \textbf{Random ordering}. Vertices are randomly ordered.
  \item \textbf{Degeneracy ordering}. Repeatedly removing a vertex of minimum degree in the remaining subgraph. The later removed vertex is ordered with a smaller number.
\end{itemize}

We next propose an improvement of the coloring algorithm, termed {\emph {dynamic coloring}}. In particular, the conflict graph will vary during the coloring process by deleting some of the existing edges based on the fact that a successfully decoded file by a WCS can be used as additional side information for later broadcast transmission in the same time slot. This results in a coloring scheme that may not be a proper coloring of the original graph but is sufficient for the transmissions to satisfy all WCS's requests.

The algorithm works as follows. Initially, we construct a conflict graph $G = (V, E)$ and order the vertices in the same way as in the single request approach. Given an ordering of vertices $v_1, v_2, ..., v_{|V|}$ of graph $G$, the dynamic greedy coloring operates across vertices as follows.

1) We start with a graph $G_1 = G$. Assign to the first vertex $v_1 = (i_1, j_1)$ color 1.

2) Consider the subgraph $G'_1$ induced by removing $v_1$ (and associated edges) from $G_1$. In addition, add file $j_1$ into WCS $i_1$'s side information set $S_{i_1}$. Note that although $v_1$ is removed from the graph, there may be other vertices representing $i_1$ (with different requested files).

3) Update $G'_1$ to a new conflict graph $G_2$ by removing edges due to the expanded side information set $S_{i_1}$. In particular, it is sufficient to check edges between vertices $v' = (i_1, j')$ and $v'' = (i'', j_1)$, namely vertices with either common WCS or common requested file with $v_1$. If file $j'$ is in the side information set of $i''$, then we can remove the edge between vertices $v'$ and $v''$.

4) Assign the first available color to vertex $v_2 = (i_2, j_2)$ and repeat the process as in steps 1), 2) and 3) by adding file $j_2$ to WCS $i_2$'s side information set $S_{i_2}$ and update the remaining graph. Then color all remaining vertices by repeating the process for vertices $v_1$ and $v_2$, until all vertices are colored.

We add the following remarks on the performance of the proposed algorithm.
\begin{remark}
Given an order of vertices $v_1, ..., v_{|V|}$ of a conflict graph $G$, the number of transmissions is at most $d+1$, where $d = \max\{d_v\}$ and $d_v = |(v, w) \in E(G), \text{where $order(w) < order(v)$}|$. This follows the fact that when we color vertex $v$, there are at most $d_v$ of its neighbors that have been colored so far.
\end{remark}

\begin{remark} Given the same ordering of vertices, the proposed dynamic coloring method often performs better than the simple reduction method. It is not hard to see that in our proposed dynamic graph coloring scheme, we may not necessary achieve a proper coloring for the original conflict graph. Indeed, we notice that two vertices corresponding to the same client with different requests are always connected, but the edge between two vertices corresponding to different clients may disappear. When we make broadcast transmissions, we still encode the files corresponding to the same color as a transmission. Therefore, if client $i$ can decode some file $j$ during the transmission process, let us say corresponding to color $k_1$, then this file $j$ can be put into $i$'s side information set and create coding opportunity for future transmissions, for example, corresponding to some color $k_2 > k_1$.
\end{remark}

\begin{figure}
  \centering
  \includegraphics[width=0.5\textwidth]{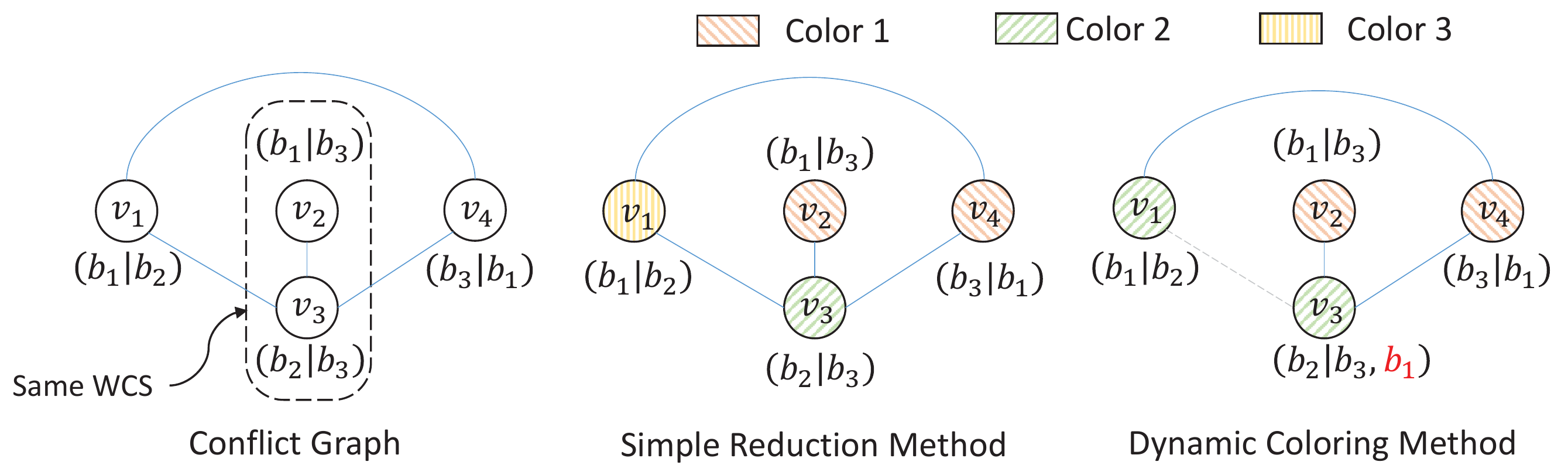}\\
  \caption{Illustration of dynamic coloring}\label{dyncolor}
\end{figure}

\textbf{Example}: \textit{To illustrate the idea of dynamic coloring and the difference between the reduction method, here we provide a simple example. Consider three WCSs and three files $\{b_1, b_2, b_3\}$. WCS 1 requests $b_1$ and has side information $b_2$; WCS 2 requests $b_1, b_2$ and has side information $b_3$; WCS 3 requests $b_3$ and has side information $b_1$. Using the graph coloring method, a conflict graph is constructed where $v_1$ represents WCS 1, $v_2$ and $v_3$ represent WCS 2, and $v_4$ represents WCS 3. The corresponding file requests and side information are annotated in the figure using the $(\cdot|\cdot)$ notation where the first entry is the requested files and the second entry is the side information. We consider a coloring order $v_2 \to v_3 \to v_4 \to v_1$. For the simple reduction method, the coloring result is illustrated in the middle figure in Fig. \ref{dyncolor}, which requires three colors. In the actual transmission phase, the MBS first broadcasts $b_1 \oplus b_3$. With their own side information, WCS 1 and WCS 3 can then obtain files  $b_3$ and $b_1$, respectively. Next the MBS broadcasts $b_2$ so WCS 2 obtains file $b_2$. Finally the MBS broadcasts $b_1$ so WCS 1 obtains file $b_1$. In fact, the simple reduction method does not save bandwidth compared to broadcasting each of the three files in three broadcast transmissions. Now consider the dynamic coloring method. We first assign color 1 to $v_2$. Then file $b_1$ is added into $v_3$'s side information set because $v_2$ and $v_3$ actually represents the same WCS. Due to this change, the edge between $v_1$ and $v_3$ is removed because now their requested files are in each other's side information set. We continue the coloring procedure and will eventually have color 1 assigned to $v_4$ and color 2 assigned to $v_3$ and $v_1$. As can been seen, the dynamic coloring method requires only two colors and hence two broadcast transmissions by the MBS. In the first transmission, the MBS broadcasts $b_1 \oplus b_3$. In the second transmission, the MBS broadcasts $b_1 \oplus b_2$.}

We have the following theorem to characterize the worst case performance of using index coding.
\begin{theorem}
If users coming to WCSs have randomly and independently distributed preference rankings across WCSs at the initial round and their preference rankings evolve independently and randomly (according to some drifting speed $c$) over rounds, then with high probability (i.e., $1-o(1)$, as $m,n,s$ tend to infinity), the number of transmissions achieved by the greedy coloring method at each round is upper bounded by $n\sqrt{c}(1-\frac{s}{n\sqrt{c}\log(s)})$.
\end{theorem}
For details of the proof, refer to the Appendix~\ref{app:thm3}.
From this theorem, we can see that the number of transmissions at each round is  proportional to the drifting speed parameter $\sqrt{c}$. We also see that there is a fraction of $\frac{s}{n\sqrt{c}\log(s)}$ transmission savings using the index coding compared with uncoded transmissions.

\section{Simulation}
The simulation setup is as follows. Each file starts with a popularity value randomly chosen in the range $[0, 1]$. The file popularity evolves over time. In each time slot, the popularity differs by a value randomly chosen in the range $[-p/2, p/2]$ compared to the popularity in the previous time slot. We call $p$ the drifting parameter, which will result in different Kendall tau distances. To capture the spatial popularity diversity, for each WCS, a $q \in [0, 1]$ fraction of randomly selected files follow a separate popularity dynamics. Therefore, if $q = 0$, then the popularity dynamics of all files are the same for all WCSs and if $q = 1$, the popularity dynamics of all files are different for all WCSs. 

Figs. \ref{performance_m} and \ref{performance_n} compare the performance of various transmission schemes when $s = 20$, $p = 0.1$ and $q = 0.2$. Fig. \ref{performance_m} investigates the impact of the number files by varying $m$ and fixing $n = 10$. Fig. \ref{performance_n} investigates the impact of the number of WCSs by varying $n$ and fixing $m = 100$. Each point is generated by running 200 time slots. As shown, coded transmission significantly reduces the number of transmissions compared to uncoded transmission: the proposed index coding-based scheme with degeneracy ordering performs the best, achieving up to 30\% bandwidth saving.

\begin{figure}
  \centering
  \includegraphics[width=0.35\textwidth]{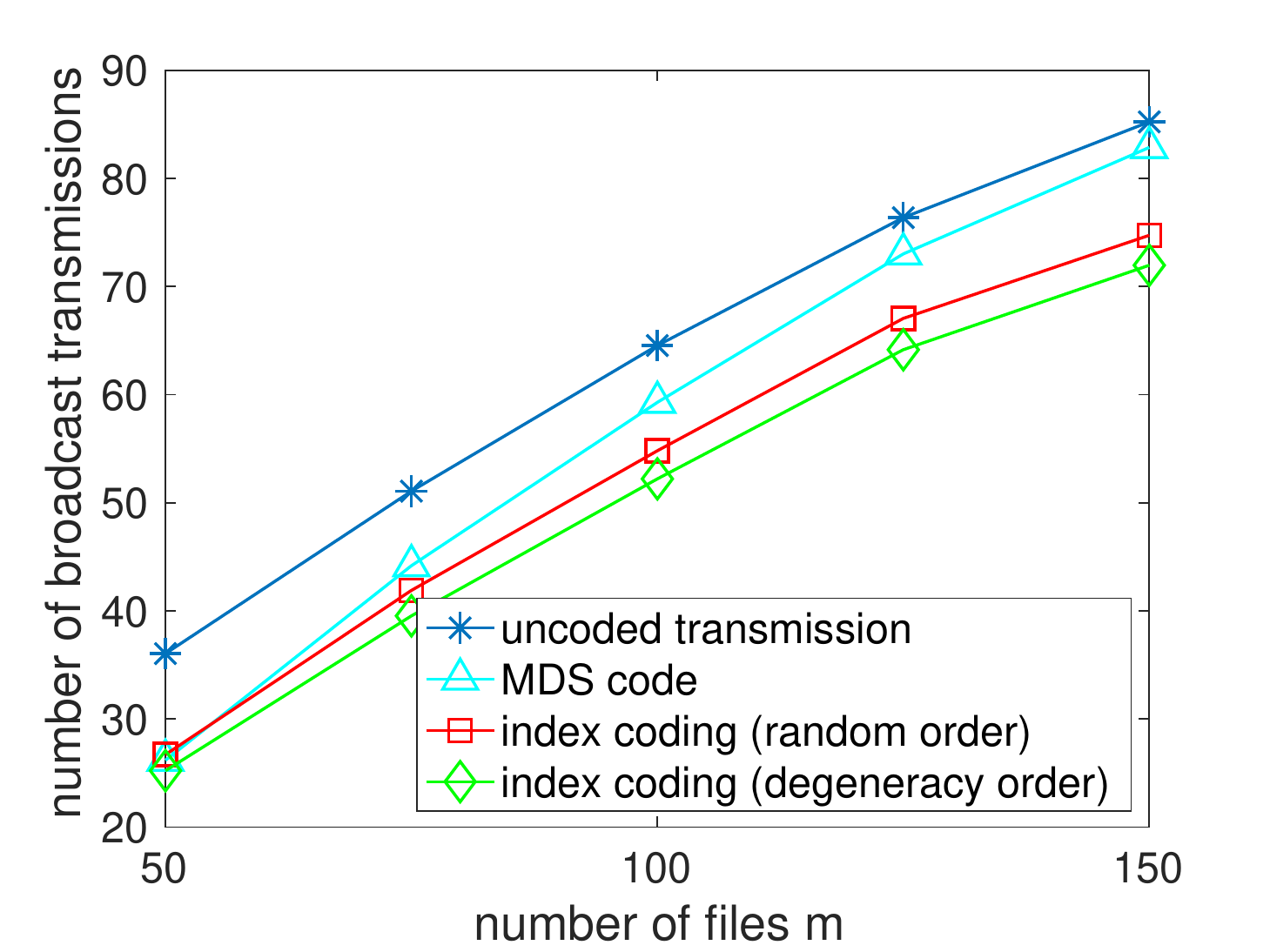}\\
  \caption{Performance comparison under different $m$}\label{performance_m}
  \vspace{-0.1in}
\end{figure}

\begin{figure}
  \centering
  \includegraphics[width=0.35\textwidth]{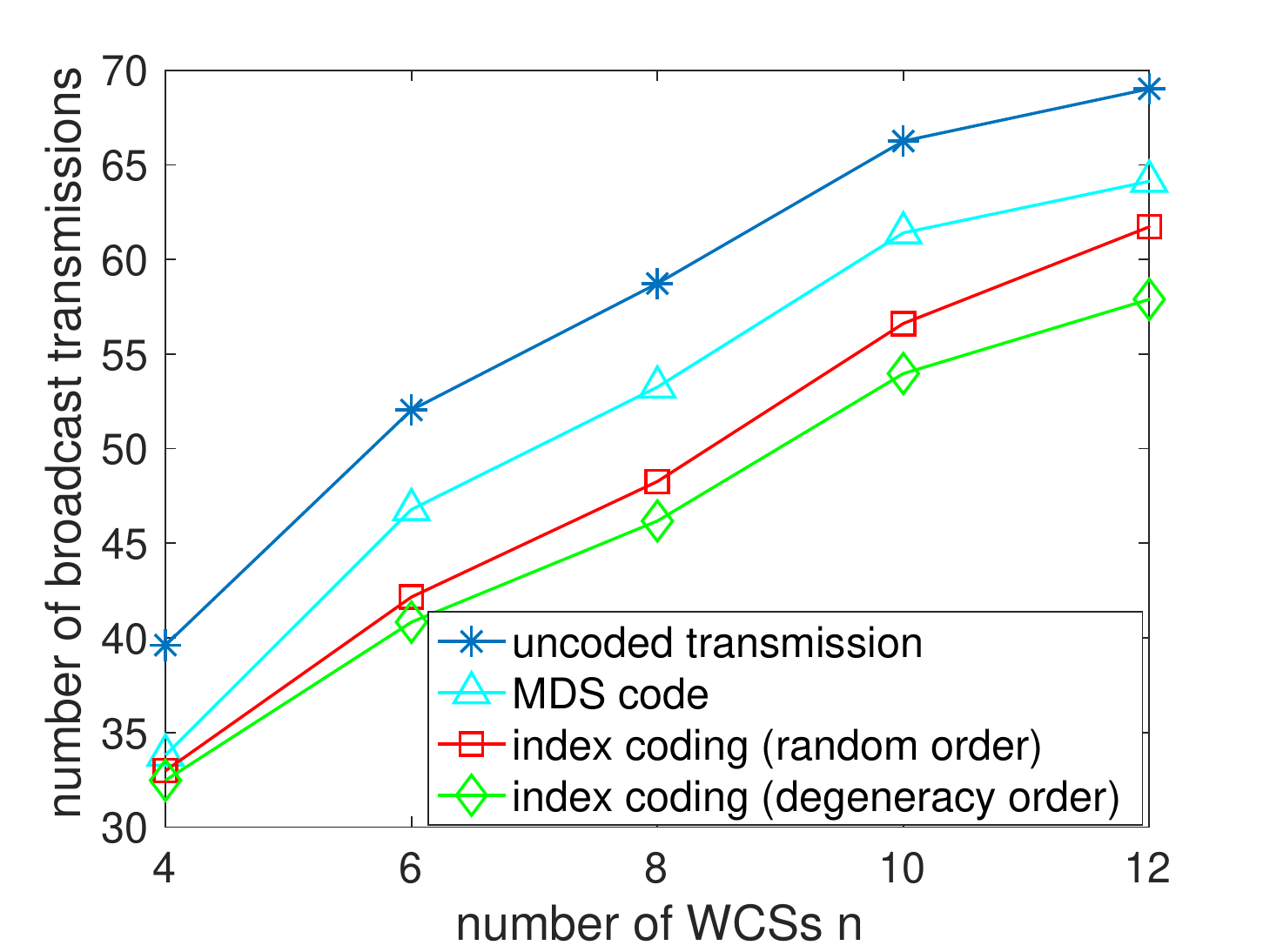}\\
  \caption{Performance comparison under different $n$}\label{performance_n}
\end{figure}

Figs. \ref{Kdist_drift} and \ref{performance_drift} illustrate the impact of popularity drifting on the system performance. Fig. \ref{Kdist_drift} shows the mean and standard deviation of the Kendall tau distance achieved under different drifting parameters $p$. Clearly, a larger $p$ results in a larger Kendall tau distance. Fig. \ref{performance_drift} shows the average number of transmissions by varying $p$. When $p$ is larger, popularity varies faster and hence more transmissions are needed to replace old files with new files. Again, the proposed index coding-based scheme with degeneracy ordering outperforms all other schemes in almost all cases. 

\begin{figure}
  \centering
  \includegraphics[width=0.35\textwidth]{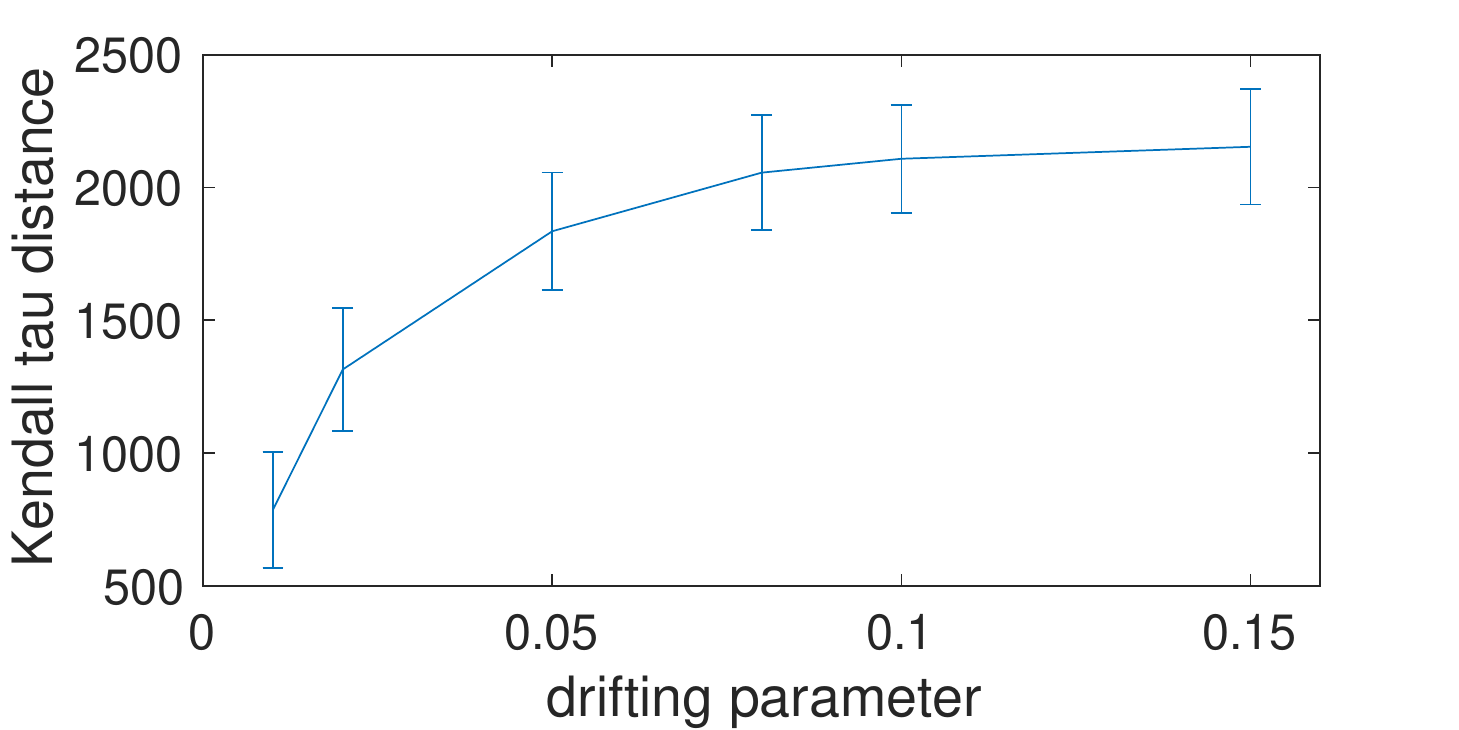}\\
  \caption{Drifting speed in terms of Kendall tau distance.}\label{Kdist_drift}
  \vspace{-0.1in}
\end{figure}

\begin{figure}
  \centering
  \includegraphics[width=0.35\textwidth]{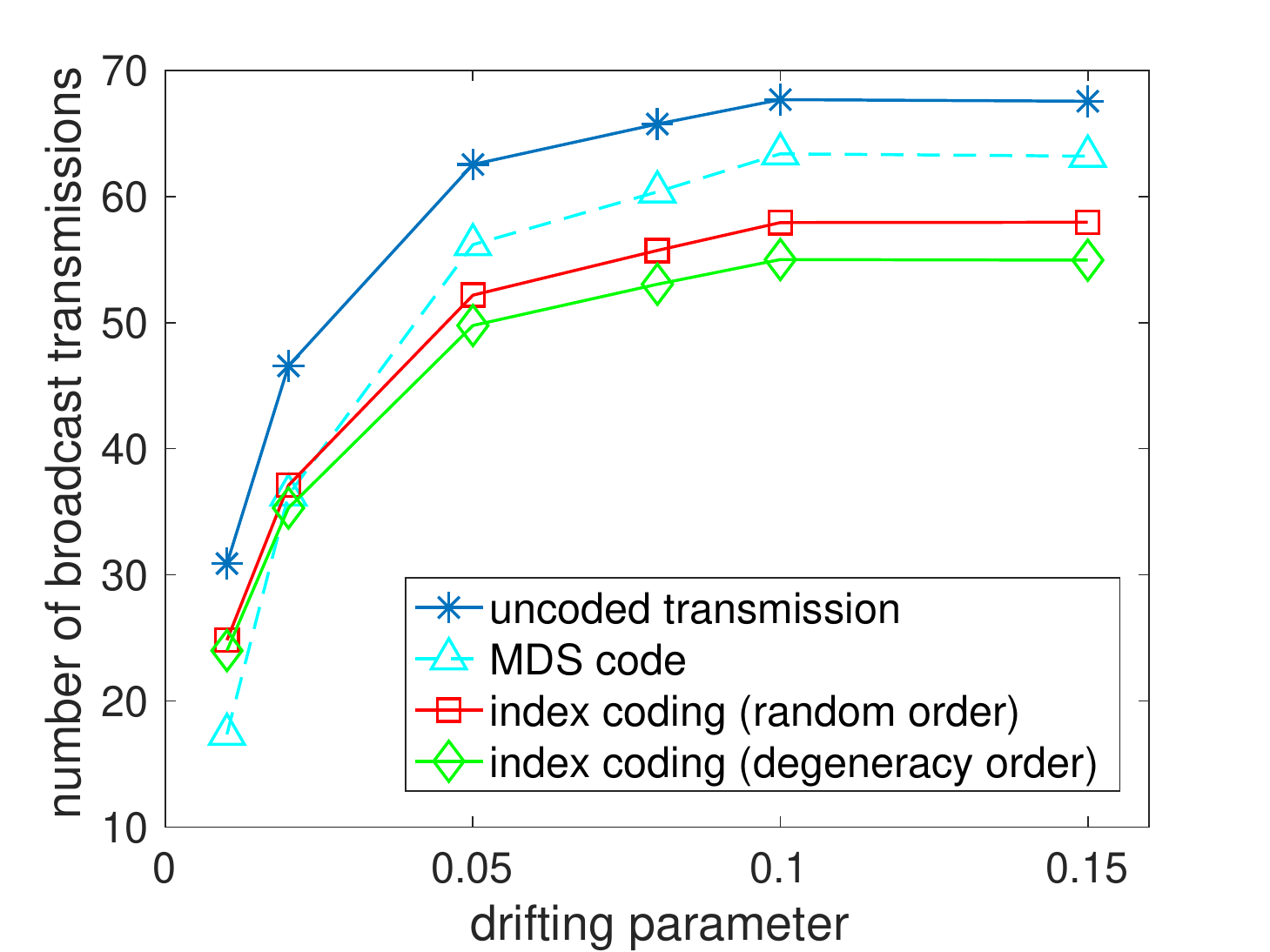}\\
  \caption{Performance comparison under different drifting speed}\label{performance_drift}
\end{figure}

\bibliographystyle{IEEEtran}
\bibliography{mybib}
\appendices
\section{Proof of Theorem 3}
\label{app:thm3}
In this appendix, we prove the relationship of the number of transmissions $T$ and the drifting speed $c$. Before describing the theorem and the proof, we first introduce two (implicit) assumptions.

We first describe a set of relationships between the parameters $m$, $n$, $c$, and $s$. Note that then notation $o(\cdot)$, $O(\cdot)$, $\Omega(\cdot)$ and $\Theta(\cdot)$ are associated with the above variables. For example, $f(m) = O(m)$ means that $f(m)/m \to C < \infty$, as $m \to \infty$.

$\bullet$ {\bf Assumption $1$}: we assume that the number of files $m$ and the number of total possible requests $n\sqrt{c}$ have the same order of magnitude. Formally, we assume $m = \beta_1 n\sqrt{c}$ for $1 \le \beta_1 = \Theta(1)$. This is a reasonable assumption, because if the total possible requests is too small, say $n\sqrt{c} = o(m)$, then this requires only a $o(m)$ number of transmissions, even if we just use uncoded transmissions.   

$\bullet$ {\bf Assumption $2$}: we assume that the number of files $m$ and the size of the caches are in the same order of magnitude. Formally, we assume that $m = \beta_2 s$ for $1 \le \beta_2 = \Theta(1)$. This assumption indicates that a fraction $1/\beta_2$ of files are cached in each WCS.

We thus reiterate Theorem~3 as follows.
\begin{theorem*}
If users coming to WCSs have randomly and independently distributed preference rankings across WCSs at the initial round and their preference rankings evolve independently and randomly (according to some drifting speed $c$) over rounds, then with high probability (i.e., $1-o(1)$, as $m,n,s$ tend to infinity), the number of transmissions achieved by the greedy coloring method at each round is upper bounded by $n\sqrt{c}(1-\frac{s}{n\sqrt{c}\log(s)})$.
\end{theorem*}

\begin{proof}
To prove this theorem, we need to calculate that any of the vertices have degree at most $n\sqrt{c}(1-\frac{s}{n\sqrt{c}\log(s)})$ with high probability (WHP) i.e., $1-o(1)$ (in random graph theory, this is also called {\em almost surely}). 

Let us denote by $V$ the set of vertices and $|V| \le n\sqrt{c}$, since the number of vertices is at most the maximum possible requests of all WCSs. Let us denote by $d(i,j)$ the degree of vertex $(i,j) \in V$ that corresponds to WCS $i$ requesting file $j$. We next would like to show that 
\begin{equation}
\begin{array}{ll}
\Pr\{d(i,j) \ge n\sqrt{c}(1-\frac{s}{n\sqrt{c}\log(s)}), \forall (i,j) \in V\} \\
= o(1).
\end{array}
\end{equation}

Due to symmetry, we can bound the above probability by 
\begin{equation}
\begin{array}{ll}
\Pr\{d(i,j) \ge n\sqrt{c}(1-\frac{s}{n\sqrt{c}\log(s)}), \forall (i,j) \in V\} \\
\le |V|\Pr\{d(1,j_{11}) \ge n\sqrt{c}(1-\frac{s}{n\sqrt{c}\log(s)})\}.
\end{array}
\end{equation}

Hence, we only need to calculate the probability that the degree of a specified vertex $d(1,j_{11})$ is above $n\sqrt{c}(1-\frac{s}{n\sqrt{c}\log(s)})$. Or equivalently, we denote by $L$ the {\em non-connection degree} of vertex $d(1,j_{11})$ with respect to the maximum possible $n\sqrt{c}$ vertices. Formally, we count $L = L_1+L_2+L_3$ in the following three cases:

$\bullet$ If the number of vertices $|V|$ is less than the maximum possible number of $n\sqrt{c}$, we count the difference as $L_1$, i.e., $L_1 = n\sqrt{c} -|V|$.

$\bullet$ If a vertex $(i,j) \in V$ requests the same file as $(1,j_11)$, i.e., $j = j_{11}$, then we count these number of vertices as $L_2$. Obviously, there is no edge between such a vertex $(i,j)$ and $(1,j_{11})$ according to our index coding conflict graph construction.

$\bullet$ If a vertex $(i,j) \in V$ and the vertex $(1,j_{11})$ have the following caching pattern: $j \in S_{1}$ and $j_{11} \in S_{i}$, then we count these number of vertices $(i,j)$ as $L_3$.

Obviously, the degree of $(1,j_{11})$ is $n\sqrt{c} - L$, then we only need to show that $L > s /\log(s)$ WHP. And it suffices for us to show that $L \le s/\log(s)$ with probability $o(1)$.

To see this, we define the following events $E_2 \triangleq \{j_{11} \in S_2\}$, $E_3 \triangleq \{j_{11} \in S_3\}$, $\ldots$, $E_n \triangleq \{j_{11} \in S_n\}$. We also define the random variables $X_2 \triangleq I_{\{E_2\}}$, $X_3 \triangleq I_{\{E_3\}}$, $\ldots$, $X_n \triangleq I_{\{E_n\}}$, where $I_{\{\}}$ is the indicator function. We let $X = X_2 + X_3 + \ldots + X_n$. Because of the independent preference rankings across WCSs, we can see that the random variables $X_i$, $i = 2, 3, \ldots, n$, are linearly independent. We thus calculate the expectation of $X$ as follows.
\begin{equation}
\mathbb{E}X = (n-1) \Pr\{E_2\} = (n-1)s/m.
\end{equation}
Using Chernoff bound, we can bound the probability that $X \le \mathbb{E}X/2$.
\begin{equation}
\Pr\{X \le \mathbb{E}X/2 \} \le e^{-\frac{(1/2)^2(n-1)s}{2m}} = e^{-\frac{(n-1)s}{8m}}.
\end{equation}

Let us denote by $V_1$ the subset of vertices that correspond to WCSs with $X_i = 1$, i.e., $V_1 = \{(i,j) \in V | X_i = 1, i \in [2:n]\}$, where $[2:n]$ denotes the set $\{2,3,\ldots,n\}$. We then need some manipulation of $V_1$. If $|V_1| < X \sqrt{c}$, then we can add some dummy vertices into $V_1$ and these dummy vertices count for $L_1$ based on our counting of $L$. If some $(i,j) \in V_1$ has the same request as $(1,j_{11})$, i.e., $j = j_{11}$, we can also replace these vertices by dummy vertices, since these vertices also counts for $L$. Hence, we can see that the worst case is that $|V_1| = X \sqrt{c}$ and the requested files are different, i.e., $|\{j|(i,j) \in V_1 \text{for some $i$}\}| = |V_1|$. Define $n_1 = \frac{(n-1)s\sqrt{c}}{2m}$ and $n_2 = m - r_1 - n_1$. Then we can see that
\begin{equation}
\begin{array}{ll}
\Pr\{d(1,j_{11}) \ge n\sqrt{c}(1-\frac{s}{n\sqrt{c}\log(s)})\} \\
\le e^{-\frac{(n-1)s}{8m}} + \sum_{k = 0}^{s/\log(s)}\frac{{n_1 \choose k}{n_2 \choose s-k}}{{n_1 + n_2 \choose s}},
\end{array}
\end{equation}
where the first term in the last expression corresponds to the probability that $X \le \mathbb{E}X/2$ and the second term corresponds to the worst case (i.e., $|V_1| = X \sqrt{c} = |\{j|(i,j) \in V_1 \text{for some $i$}\}|$) probability that for the $V_1$ vertices, the non-connection degree of $(1,j_{11})$ is below $s/\log(s)$.

We then calculate that 
\begin{equation}
\begin{array}{ll}
\sum_{k = 0}^{s/\log(s)}\frac{{n_1 \choose k}{n_2 \choose s-k}}{{n_1 + n_2 \choose s}} \\
\le \sum_{k = 0}^{s/\log(s)}\frac{n_1^k n_2^{(s-k)} 4s!}{k!(s-k)!(n_1 + n_2)^s} \\
\le \frac{4n_2^s}{(n_1+n_2)^s}\sum_{k = 0}^{s/\log(s)} {s \choose k} \\
\le \frac{4n_2^s}{(n_1+n_2)^s} (\frac{es}{s/\log(s)})^{s/\log(s)} \\
\le 4e^{s\log(1-\frac{n_1}{n_1+n_2}) + \frac{s}{log(s)}(1+\log\log(s))} \\
\le 4 e^{-s(\frac{1}{8\beta_1\beta_2}-o(1))} 
\end{array}
\end{equation}
where the first inequality comes from $\frac{x^y}{4y!} \le {x \choose y} \le \frac{x^y}{y!}$; the second inequality follows from that $n_1 = \frac{(n-1)s\sqrt{c}}{2m} \le \frac{ns\sqrt{c}}{2m} = \frac{m}{2 \beta_1 \beta_2} \le (n_1+n_2)/2 \le n_2$; the third inequality follows from that $\sum_{k = 0}^{d} {s \choose k} \le (es/d)^d$; and the last inequality follows from that $\log(1-x) \le -\frac{1}{2}x$ for $x \le 0.2$.

We then use the union bound to bound the probability that no vertices have degree larger than $n\sqrt{c}(1-\frac{s}{n\sqrt{c}\log(s)})$.
\begin{equation}
\begin{array}{ll}
\Pr\{d(i,j) \ge n\sqrt{c}(1-\frac{s}{n\sqrt{c}\log(s)}), \forall (i,j) \in V\} \\
\le n\sqrt{c} [e^{-\frac{(n-1)s}{8m}}+4 e^{-s(\frac{1}{8\beta_1\beta_2}-o(1))}] \\
\le e^{-n(\frac{1}{16\beta_2}-o(1))} + e^{-s(\frac{1}{8\beta_1\beta_2}-o(1))} = o(1).
\end{array}
\end{equation}
Then the result follows from that the number of colors in greedy coloring method is upper bounded by $n\sqrt{c}(1-\frac{s}{n\sqrt{c}\log(s)}) + 1 \approx n\sqrt{c}(1-\frac{s}{n\sqrt{c}\log(s)})$.
\end{proof}

\end{document}